\documentclass{article}

\usepackage{fullpage}
\usepackage{amsthm,amssymb,amsmath}  
\usepackage{comment,xspace,enumerate}
\usepackage[capitalise]{cleveref}
\usepackage[utf8]{inputenc}
\usepackage{thmtools}
\usepackage{thm-restate}
\usepackage{authblk}

\newcommand{\myparagraph}[1]{\paragraph{#1.}}

  \theoremstyle{plain}
  \newtheorem{theorem}{Theorem}[section]
  \newtheorem{lemma}[theorem]{Lemma}  
  \newtheorem{corollary}[theorem]{Corollary}  
  \newtheorem{fact}[theorem]{Fact}
  \newtheorem{observation}[theorem]{Observation}
  \newtheorem{lem}[theorem]{Lemma}
  \renewenvironment{lemma}{\begin{lem}}{\end{lem}}
  \crefname{lem}{Lemma}{Lemmas}
  \newtheorem{cor}[theorem]{Corollary}
  \renewenvironment{corollary}{\begin{cor}}{\end{cor}}
  \crefname{cor}{Corollary}{Corollaries}
  \theoremstyle{definition}
  \newtheorem{definition}[theorem]{Definition}
   \newtheorem{defi}[theorem]{Definition}
  \renewenvironment{definition}{\begin{defi}}{\end{defi}}
  \crefname{defi}{Definition}{Definitions}
 
  \newtheorem*{claim}{Claim}

   \title{Minimal Suffix and Rotation of a Substring in Optimal Time\footnote{This work is supported by Polish budget funds for science in 2013-2017 as a research project under the `Diamond Grant' program.}}

\author{Tomasz Kociumaka}

\affil{Institute of Informatics, University of Warsaw, Poland\\
    \texttt{kociumaka@mimuw.edu.pl}}
    
\date{\vspace{-5ex}}

\newsavebox{\mybox}
\newenvironment{problem}[1]
{\begin{center}\begin{lrbox}{\mybox}\begin{minipage}{0.96\columnwidth}#1\\}
{\end{minipage}\end{lrbox}\fbox{\usebox{\mybox}}\end{center}}

\newcommand{\AMSQ}{\textsc{Auxiliary Minimal Suffix Queries}\xspace}
\newcommand{\MSQ}{\textsc{Minimal Suffix Queries}\xspace}
\newcommand{\GMSQ}{\textsc{Generalized Minimal Suffix Queries}\xspace}
\newcommand{\MRQ}{\textsc{Minimal Rotation Queries}\xspace}

\newcommand{\dol}{\$}
\newcommand{\eps}{\varepsilon}
\newcommand{\Oh}{\mathcal{O}}
\newcommand{\pred}{\mathrm{pred}}
\newcommand{\suc}{\mathrm{succ}}
\newcommand{\rank}{\mathrm{rank}}
\newcommand{\select}{\mathrm{select}}
\newcommand{\sub}{\subseteq}

\newcommand{\T}{\mathcal{T}}
\newcommand{\floor}[1]{\left\lfloor#1\right\rfloor}
\newcommand{\ceil}[1]{\left\lceil#1\right\rceil}
\newcommand{\Sig}{\Lambda}
\newcommand{\MinSuf}{\mathrm{MinSuf}}

\newcommand{\MaxSuf}{\mathrm{MaxSuf}}
\newcommand{\RMaxSuf}{\mathrm{MaxSuf}^R}
\newcommand{\RMinSuf}{\mathrm{MinSuf}^R}
\newcommand{\rev}{\prec^R}
\newcommand{\lcp}{\mathrm{lcp}}
\newcommand{\lcs}{\mathrm{lcs}}

\newcommand{\A}{\mathcal{A}}
\newcommand{\oid}{\mathrm{oid}}
\newcommand{\esigma}{\bar{\Sigma}}

\begin{document}

\maketitle

\begin{abstract}
For a text given in advance, the substring minimal suffix queries ask to determine the lexicographically minimal non-empty suffix of a substring specified by the location of its occurrence in the text.
We develop a data structure answering such queries optimally: in constant time after linear-time preprocessing.
This improves upon the results of Babenko et al.~(CPM 2014), whose trade-off solution is characterized by $\Theta(n\log n)$ product of these time complexities.
Next, we extend our queries to support concatenations of $\Oh(1)$ substrings, for which the construction and query time is preserved.
We apply these generalized queries to compute lexicographically minimal and maximal rotations of a given substring in constant time after linear-time preprocessing.

Our data structures mainly rely on properties of Lyndon words and Lyndon factorizations.
We combine them with further algorithmic and combinatorial tools, such as fusion trees and the notion of order isomorphism of strings. 
\end{abstract}

\section{Introduction}

Lyndon words, as well as the inherently linked concepts of the lexicographically minimal suffix 
and the lexicographically minimal rotation of a string, are one of the most successful concepts of combinatorics of words.
Introduced by Lyndon~\cite{Lyndon1954} in the context of Lie algebras, they are widely used in algebra and combinatorics.
They also have surprising algorithmic applications, including ones related to constant-space pattern matching~\cite{ConstantSpacePM}, maximal repetitions~\cite{B14bis}, and the shortest common superstring problem~\cite{DBLP:conf/soda/Mucha13}.

The central combinatorial property of Lyndon words, proved by Chen et al.~\cite{chen1958free}, states
that every string can be uniquely decomposed into a non-increasing sequence of Lyndon words. 
Duval~\cite{Duval} devised a simple algorithm computing the Lyndon factorization in linear time and constant space.  
He also observed that the same algorithm can be used to determine the lexicographically minimal and maximal suffix,
as well as the lexicographically minimal and maximal rotation of a given string.

The first two algorithms are actually on-line procedures: in linear time they allow computing the minimal and maximal suffix of every prefix of a given string. 
For rotations such a procedure was later introduced by Apostolico and Crochemore~\cite{DBLP:journals/iandc/ApostolicoC91}. 
Both these solutions lead to the optimal, quadratic-time algorithms computing the minimal and maximal suffixes and rotations for all substring of a given string. 
Our main results are the data-structure versions of these problems: we preprocess a given text $T$ to answer the following queries:

\begin{problem}{\MSQ}
Given a substring $v=T[\ell..r]$ of $T$, report the lexicographically smallest non-empty suffix of $v$ (represented by its length).
\end{problem}
\begin{problem}{\MRQ}
Given a substring $v=T[\ell..r]$ of $T$, report the lexicographically smallest rotation of $v$ (represented by the number of positions to shift).
\end{problem}

\noindent
For both problems we obtain optimal solutions with linear construction time and constant query time.
For \MSQ this improves upon the results of Babenko et al.~\cite{Babenko2016}, who developed a trade-off solution,
which for a text of length $n$ has $\Theta(n\log n)$ product of preprocessing and query time. 
We are not aware of any results for \MRQ except for a data structure only testing cyclic equivalence of two subwords~\cite{InternalPM}.
It allows constant-time queries after randomized preprocessing running in expected linear time.

An optimal solution for the \textsc{Maximal Suffix Queries} was already obtained in~\cite{Babenko2016},
while the \textsc{Maximal Rotation Queries} are equivalent to \MRQ subject to alphabet reversal. 
Hence, we do not focus on the maximization variants of our problems.

Using an auxiliary result devised to handle \MRQ, we also develop a data structure answering in $\Oh(k^2)$ time the following generalized queries:

\begin{problem}{\GMSQ}
Given a sequence of substrings $v_1,\ldots,v_k$ ($v_i=T[\ell_i..r_i]$), report the lexicographically smallest
non-empty suffix of their concatenation $v_1v_2\ldots v_k$ (represented by its length).\end{problem}

All our algorithms are deterministic procedures for the standard word RAM model with machine words of size $W=\Omega(\log n)$~\cite{DBLP:conf/stacs/Hagerup98}. 
The alphabet is assumed to be $\Sigma=\{0,\ldots,\sigma-1\}$ where $\sigma=n^{\Oh(1)}$, so that all letters of the input text $T$ can be sorted in linear time. 
 
\myparagraph{Applications} The last factor of the Lyndon factorization of a string is its minimal suffix.
As noted in~\cite{Babenko2016}, this can be used to reduce computing the factorization $v=v_1^{p_1}\cdots v_m^{p_m}$
of a substring $v=T[\ell..r]$ to $\Oh(m)$ \MSQ in $T$. Hence, our data structure determines the factorization in the optimal $\Oh(m)$ time.
If $v$ is a concatenation of $k$ substrings, this increases to $\Oh(k^2 m)$ time (which we did not attempt to optimize in this paper).

The primary use of \MRQ is \emph{canonization} of substrings, i.e., classifying them according to cyclic equivalence (conjugacy); see~\cite{DBLP:journals/iandc/ApostolicoC91}. As a proof-of-concept application of this natural tool, we propose counting distinct substring with a given exponent.

\myparagraph{Related work}
Our work falls in a class of substring queries: data structure problems solving basic stringology problems for substrings of a preprocessed text. This line of research, implicitly initiated by substring equality and longest common prefix queries 
(using suffix trees and suffix arrays; see~\cite{AlgorithmsOnStrings}),
now includes several problems related to compression~\cite{SubstringCompression,GeneralizedSubstringCompression,InternalPM,WaveletSuffixTree},
pattern matching~\cite{InternalPM}, and the range longest common prefix problem~\cite{RangeLCP,FasterRangeLCP,DBLP:conf/spire/AmirLT15}.  
Closest to ours is a result by Babenko et al.~\cite{WaveletSuffixTree}, which after $\Oh(n\sqrt{\log n})$-expected-time preprocessing allows determining the $k$-th smallest suffix of a given  substring, as well as finding the lexicographic rank of one substring
among suffixes of another substring, both in logarithmic time

\myparagraph{Outline of the paper}
In \cref{sec:prelim} we recall standard definitions and two well-known data structures.
Next, in \cref{sec:comb}, we study combinatorics of minimal suffixes, using in particular a notion of \emph{significant suffixes}, introduced by I et al.~\cite{DBLP:conf/cpm/INIBT13,DBLP:conf/spire/INIBT13} to compute Lyndon factorizations of grammar-compressed strings. 
\cref{sec:msq} is devoted to answering \MSQ. 
We use \emph{fusion trees} by Pătraşcu and Thorup~\cite{DBLP:conf/focs/PatrascuT14} to improve the query time from logarithmic to $\Oh(\log^* |v|)$, and then, by preprocessing shorts strings, we achieve constant query time.
That final step uses a notion of \emph{order-isomorphism}~\cite{DBLP:journals/ipl/KubicaKRRW13,DBLP:journals/tcs/KimEFHIPPT14} to reduce
the number of precomputed values. In \cref{sec:gmsq} we repeat the same steps for \GMSQ.
We conclude with \cref{sec:app}, where we briefly discuss the~applications.

\section{Preliminaries}\label{sec:prelim}
We consider strings over an alphabet $\Sigma=\{0,\ldots,\sigma-1\}$ with the natural order $\prec$.
The empty string is denoted as $\eps$. 
By $\Sigma^*$ ($\Sigma^+$) we denote the set of all (resp. non-empty) finite strings over $\Sigma$.
We also define $\Sigma^\infty$ as the set of infinite strings over $\Sigma$.
We extend the order $\prec$ on $\Sigma$ in the standard way to the \emph{lexicographic} order on $\Sigma^{*}\cup \Sigma^\infty$.

Let $w=w[1]\ldots w[n]$ be a string in $\Sigma^*$.
We call $n$ the \emph{length} of $w$ and denote it by $|w|$.
For $1\le i \le j \le n$, a string $u=w[i] \ldots w[j]$ is called a \emph{substring} of $w$.
By $w[i..j]$ we denote the occurrence of $u$ at position $i$, called a \emph{fragment} of~$w$.
A fragment of $w$ other than the whole $w$ is called a \emph{proper} fragment of $w$.
A fragment starting at position $1$ is called a \emph{prefix} of $w$
and a fragment ending at position $n$ is called a \emph{suffix} of $w$.
We use abbreviated notation $w[..j]$ and $w[i..]$ for a prefix $w[1..j]$ and a suffix $w[i..n]$ of $w$,
respectively. A \emph{border} of $w$ is a substring of $w$ which occurs both as a prefix and as a suffix of $w$.
An integer $p$, $1\le p \le |w|$, is a \emph{period} of $w$ if $w[i]=w[i+p]$ for $1\le i \le n-p$.
If $w$ has period $p$, we also say that is has \emph{exponent} $\frac{|w|}{p}$.
Note that $p$ is a period of $w$ if and only if $w$ has a border of length $|w|-p$.

We say that a string $w'$ is a \emph{rotation} (cyclic shift, conjugate) of a string $w$ if there exists a decomposition $w=uv$ such that $w'=vu$. 
Here, $w'$ is the left rotation of $w$ by $|u|$ characters and the right rotation of $w$ by $|v|$ characters.

\myparagraph{Enhanced suffix array}
The \emph{suffix array}~\cite{DBLP:journals/siamcomp/ManberM93} of a text $T$ of length $n$ is a permutation $SA$ of $\{1,\ldots,n\}$
defining the lexicographic order on suffixes $T[i..n]$: $SA[i]<SA[j]$ if and only if $T[i..n]\prec T[j..n]$. 
For a string $T$, both $SA$ and its inverse permutation $ISA$ take $\Oh(n)$ space and can be computed in $\Oh(n)$ time; see e.g.~\cite{AlgorithmsOnStrings}.
Typically, one also builds the $LCP$ table and extends it with a data structure for range minimum queries~\cite{DBLP:journals/siamcomp/HarelT84,LCA},
so that the longest common prefix of any two suffixes of $T$ can be determined efficiently.

Similarly to~\cite{Babenko2016}, we also construct these components for the reversed text $T^R$.
Additionally, we preprocess the $ISA$ table to answer range minimum and maximum queries.
The resulting data structure, which we call the \emph{enhanced suffix array} of $T$,
lets us perform many queries.

\begin{theorem}[Enhanced suffix array; see Fact~3 and Lemma~4 in~\cite{Babenko2016}]\label{thm:esa}
The enhanced suffix array of a text $T$ of length~$n$ takes $\Oh(n)$ space,
can be constructed in $\Oh(n)$ time, and allows answering the following queries in $\Oh(1)$ time
given fragments $x$, $y$ of $T$:
\begin{enumerate}[(a)]
  \item\label{it:cmp} determine if $x\prec y$, $x=y$, or $x\succ y$,
  \item\label{it:lcps} compute the the longest common prefix $\lcp(x,y)$ and the longest common suffix $\lcs(x,y)$,
  \item\label{it:lcpex} compute $\lcp(x^\infty,y)$ and determine if $x^{\infty} \prec y$, $x^\infty=y$, or $x^\infty\succ y$.
 \end{enumerate}
 Moreover, given indices $i,j$, it can compute in $\Oh(1)$ time the \emph{minimal} and the \emph{maximal} suffix among $\{T[k..n] :i\le k \le j\}$.
\end{theorem}

\myparagraph{Fusion trees} Consider a set $\A$ of $W$-bit integers (recall that $W$ is the machine word size).
\emph{Rank} queries given a $W$-bit integer $x$ return $\rank_{\A}(x)$ defined
as $|\{y\in \A : y < x\}|$. Similarly, \emph{select} queries given an integer $r$, $0\le r < |\A|$,
return $\select_{\A}(r)$, the $r$-th smallest element in $\A$, i.e., $x\in \A$ such that $\rank_{\A}(x)=r$.
These queries can be used to determine the \emph{predecessor} and the \emph{successor} of a $W$-bit integer $x$,
i.e., $\pred_{\A}(x)=\max\{y \in \A : y < x\}$ and $\suc_{\A}(x)=\min\{y\in \A : y \ge x\}$.
We answer these queries with dynamic fusion trees by Pătraşcu and Thorup~\cite{DBLP:conf/focs/PatrascuT14}.
We only use these trees in a static setting, but the original static fusion trees by Fredman and Willard~\cite{DBLP:journals/jcss/FredmanW93}
do not have an efficient construction procedure.

\begin{theorem}[Fusion trees~\cite{DBLP:conf/focs/PatrascuT14,DBLP:journals/jcss/FredmanW93}]\label{thm:fus}
There exists a data structure of size $\Oh(|\A|)$ which answers $\rank_\A$, $\select_{\A}$, $\pred_{\A}$, and $\suc_{\A}$ queries
in $\Oh(1+\log_W |\A|)$ time. Moreover, it can be constructed in $\Oh(|\A|+|\A|\log_W|\A|)$ time.
\end{theorem}


\section{Combinatorics of minimal suffixes and Lyndon words}\label{sec:comb}
For a non-empty string $v$ the \emph{minimal suffix} $\MinSuf(v)$ is the lexicographically smallest non-empty suffix $s$ of $v$.
Similarly, for an arbitrary string $v$ the \emph{maximal suffix} $\MaxSuf(v)$ is the lexicographically largest suffix $s$ of $v$.
We extend these notions as follows: for a pair of strings $v,w$ we define $\MinSuf(v,w)$ and $\MaxSuf(v,w)$
as the lexicographically smallest (resp. largest) string $sw$ such that $s$ is a (possibly empty) suffix of $v$.

In order to relate minimal and maximal suffixes, we introduce the \emph{reverse} order $\rev$ on $\Sigma$ and extend it to the \emph{reverse lexicographic} order,
and an auxiliary symbol  $ \$ \notin \Sigma$. We extend the order $\prec$ on $\Sigma$
so that $\mathtt{c} \prec \$$ (and thus $\dol \prec^R \mathtt{c}$) for every $\mathtt{c}\in \Sigma$.
We define $\esigma=\Sigma\cup \{\$\}$, but unless otherwise stated, we still assume that the strings considered belong to $\Sigma^*$.
\begin{observation}
If $u,v\in\Sigma^*$, then $u\$ \prec  v$ if and only if $v \prec^R u$.
\end{observation}

We use $\RMinSuf$ and $\RMaxSuf$ to denote the minimal (resp. maximal) suffix with respect to $\rev$.
The following observation relates the notions we introduced:
\begin{observation}\label{obs:minmax}
\begin{enumerate}[(a)]
  \item $\MaxSuf(v,\eps)=\MaxSuf(v)$ for every $v\in \esigma^*$,
    \item\label{it:minmax:c} $\MinSuf(vw) = \min(\MinSuf(v,w),\MinSuf(w))$ for every $v\in \esigma^*$ and $w\in \esigma^+$,
  \item\label{it:minmax:b} $\MinSuf(vc) = \MinSuf(v,c)$ for every $v\in \esigma^*$ and $c\in \esigma$,
    \item\label{it:minmax:d} $\MinSuf(v,w\$)=\RMaxSuf(v,w)\$$ for every $v,w\in \Sigma^*$,
    \item\label{it:minmax:e} $\MinSuf(v\dol)=\RMaxSuf(v)\dol$ for every $v\in \Sigma^*$.
  \end{enumerate}
\end{observation}
A property seemingly similar to~(\ref{it:minmax:e}) is false for every $v\in \Sigma^+$: $\dol = \RMinSuf(v\$)\ne \MaxSuf(v)\$$.

A notion deeply related to minimal and maximal suffixes is that of a Lyndon word~\cite{Lyndon1954,chen1958free}.
A string $w\in \Sigma^+$ is called a \emph{Lyndon word} if $\MinSuf(w)=w$.
Note that such $w$ does not have proper borders, since a border would be a non-empty suffix smaller than $w$.
A \emph{Lyndon factorization}  of a string $u\in \esigma^*$ is a representation $u=u_1^{p_1}\ldots u_m^{p_m}$, where $u_i$ are Lyndon words
such that $u_1 \succ \ldots \succ u_m$.
Every non-empty word has a unique Lyndon factorization~\cite{chen1958free}, which can be computed in linear time and constant space~\cite{Duval}.
The following result provides a characterization of the Lyndon factorization of a concatenation of two strings:

\begin{lemma}[\cite{ParallelLyndon,DBLP:journals/tcs/DaykinIS94}]\label{lem:lyndon_conc}
Let $u=u_1^{p_1}\cdots u_m^{p_m}$ and $v=v_1^{q_1} \cdots v_\ell^{q_\ell}$ be Lyndon factorization.
Then the Lyndon factorization of $uv$ is $uv=u_1^{p_1}\cdots u_{c}^{p_c}z^k v_{d+1}^{q_{d+1}}\cdots v_\ell^{q_\ell}$
for integers $c,d,k$ and a Lyndon word $z$ such that $0\le c < m$, $0\le d\le \ell$, and $z^k = u_{c+1}^{p_{c+1}}\cdots u_m^{p_m}v_1^{q_1}\cdots v_{d}^{q_{d}}$.
\end{lemma}
Next, we prove another simple yet useful property of Lyndon words:
\begin{fact}\label{fct:3}
Let $v,w\in \Sigma^+$ be strings such that $w$ is a Lyndon word.
If $v \prec w$, then $v^\infty \prec w$.
\end{fact}
\begin{proof}
For a proof by contradiction suppose that $v \prec w \prec v^\infty$.
Let $w=v^ks$, where $v$ is not a prefix of $s$. 
Note that $k\ge 1$ as $v$ must be a prefix of $w$. 
Because $w=v^k s\prec v^\infty$, we have $s \prec v^\infty$.
On the other hand, $w$ is a Lyndon word, so $v^k s = w \prec s$.
Consequently, $v^k s \prec s \prec v^\infty$. 
Since $k\ge 1$, $v$ must be a prefix of $s$, which contradicts the definition of $k$.
\end{proof} 

\subsection{Significant suffixes}
Below we recall a notion of \emph{significant suffixes}, introduced by I et al.~\cite{DBLP:conf/cpm/INIBT13,DBLP:conf/spire/INIBT13} in order to compute Lyndon factorizations of grammar-compressed strings. 
Then, we state combinatorial properties of significant suffixes; some of them are novel
and some were proved in \cite{DBLP:conf/spire/INIBT13}.

\begin{definition}[see \cite{DBLP:conf/cpm/INIBT13,DBLP:conf/spire/INIBT13}]\label{def:sig}
A suffix $s$ of a string $v\in \Sigma^*$ is a \emph{significant suffix} of $v$ if $sw = \MinSuf(v,w)$ for some $w\in \esigma^*$.
\end{definition}
Let $v=v_1^{p_1}\ldots v_{m}^{p_{m}}$ be the Lyndon factorization of a string $v\in \Sigma^+$.
For $1\le j \le m$ we denote  $s_j=v_j^{p_j}\cdots v_{m}^{p_{m}}$;
moreover, we assume $s_{m+1}=\eps$.
Let $\lambda$ be the smallest index such that $s_{i+1}$ is a prefix of $v_i$ for $\lambda \le i \le m$.
Observe that $s_{\lambda} \succ \ldots \succ s_m \succ s_{m+1}=\eps$, since $v_i$ is a prefix of $s_i$.
We define $y_i$ so that $v_{i} = s_{i+1}y_i$, and we set $x_i = y_is_{i+1}$. Note that $s_i = v_i^{p_i} s_{i+1}=(s_{i+1}y_{i})^{p_i}s_{i+1}=s_{i+1}(y_is_{i+1})^{p_i}=s_{i+1}x_i^{p_i}$.
We also denote $\Sig(w)=\{s_\lambda,\ldots,s_m,s_{m+1}\}$, $X(w)=\{x_\lambda^{\infty},\ldots,x_m^\infty\}$, and $X'(w)=\{x_\lambda^{p_\lambda},\ldots,x_m^{p_m}\}$. The observation below lists several immediate properties of the introduced strings:

\begin{observation}\label{obs:simple}
For each $i$, $\lambda \le i \le m$: (a) $x_i^\infty \succ x_i^{p_i}\succeq x_i \succeq y_i$, (b) $x_i^{p_i}$ is a suffix of $v$ of length $|s_{i}|-|s_{i+1}|$,
and (c) $|s_{i}| > 2|s_{i+1}|$. In particular, $|\Sig(v)|= \Oh(\log |v|)$.
\end{observation}

The following lemma shows that $\Sig(v)$ is equal to the set of significant suffixes of $v$.
(Significant suffixes are actually defined in~\cite{DBLP:conf/spire/INIBT13} as $\Sig(v)$
and only later proved to satisfy our \cref{def:sig}.)
In fact, the lemma is much deeper; in particular, the formula for $\MaxSuf(v,w)$ is one of the key ingredients of our efficient algorithms answering \MSQ.

\begin{lemma}[I et al.~\cite{DBLP:conf/spire/INIBT13}, Lemmas~12--14]\label{lem:char}
For a string $v\in \Sigma^+$ let $s_i$, $\lambda$, $x_i$, and $y_i$, be defined as above. Then
$x_{\lambda}^\infty \succ x_\lambda^{p_\lambda} \succeq y_\lambda \succ x_{\lambda+1}^\infty \succ x_{\lambda+1}^{p_{\lambda+1}}  \succeq y_{\lambda+1}\succ \ldots \succ x_m^\infty \succ x_m^{p_m} \succeq y_m.$
Moreover, for every string $w\in \esigma^*$ we have
  $$\MinSuf(v,w)=\begin{cases}
  s_\lambda w &\text{if }w\succ x_{\lambda}^\infty,\\
  s_i w & \text{if }x_{i-1}^\infty\succ w \succ x_i^\infty\text{ for }\lambda < i \le m,\\
  s_{m+1}w & \text{if }x_{m}^\infty\succ w.\\
  \end{cases}$$
 In other words, $\MinSuf(v,w)=s_{m+1-r}w$ where $r=\rank_{X(v)}(w)$.
\end{lemma}

We apply \cref{lem:char} to deduce several properties of the set $\Sig(v)$ of significant suffixes.

\begin{corollary}\label{cor:char}
For every string $v\in \Sigma^+$:
\begin{enumerate}[(a)]
  \item\label{it:rmax} the largest suffix in $\Sig(v)$ is $\RMaxSuf(v)$ and $\Sig(v)=\Sig(\RMaxSuf(v))$,
  \item\label{it:suf} if $s$ is a suffix of $v$ such that $|v|\le 2|s|$, then $\Sig(v)\subseteq \Sig(s)\cup \{\RMaxSuf(v)\}$.
\end{enumerate} 
\end{corollary}
\begin{proof}
To prove (\ref{it:rmax}), observe that $x_{\lambda} \in \Sigma^+$, so $ x_{\lambda}^\infty\prec \dol$.
Consequently, \cref{lem:char} states that $s_{\lambda}\$ = \MinSuf(v,\$)$.
However, we have $\MinSuf(v,\$)=\RMaxSuf(v)\$$ by \cref{obs:minmax}(\ref{it:minmax:e}),
and thus $s_{\lambda}=\RMaxSuf(v)$. 
Uniqueness of the Lyndon factorization implies that $u_{\lambda}^{p_{\lambda}}\cdots u_m^{p_m}$ is the Lyndon factorization of $s_{\lambda}$, 
and hence by definition of $\Sig(\cdot)$ we have $\Sig(v)=\Sig(s_{\lambda})$.

For a proof of (\ref{it:suf}), we shall show that for $i\ge \lambda+1$ the string $s_{i}$ is a significant suffix of $s$.
Note that, by \cref{obs:simple}, $s_i$ is a suffix of $s$, since $2|s_{i}|< |s_{i-1}|\le |s_{\lambda}|\le |v|\le 2|s|$.
The suffix $s_{m+1}=\eps$ is clearly a significant suffix of $s$, so we assume $\lambda < i\le m$.
By \cref{lem:char}, one can choose $w\in \esigma^*$ (setting $x_{i-1}^\infty \succ w \succ x_i^\infty$) so that $s_{i}w = \MinSuf(v,w)$.
However, this also implies $s_{i}w = \MinSuf(s,w)$ because all suffixes of $s$ are suffixes of $v$.
Consequently, $s_{i}$ is a significant suffix of $s$, as claimed.
\end{proof}

Below we provide a precise characterization of $\Sig(uv)$ for $|u|\le |v|$ in terms of $\Sig(v)$ and $\RMaxSuf(u,v)$. 
This is another key ingredient of our data structure, in particular letting us efficiently compute
significant suffixes of a given fragment of $T$.

\begin{lemma}\label{lem:extend}
Let $u,v\in \Sigma^+$ be strings such that $|u|\le |v|$. 
Also, let $\Sig(v)=\{s_{\lambda},\ldots,s_{m+1}\}$, $s'=\RMaxSuf(u,v)$,
and let $s_i$ be the longest suffix in $\Sig(v)$ which is a prefix of $s'$.
Then
$$\Sig(uv)=\begin{cases}
\{s_{\lambda},\ldots,s_{m+1}\} &\text{if }s'\preceq^R s_{\lambda}\text{ }(\text{i.e., if }s_\lambda \preceq s'\text{ and }i \ne \lambda), \\
\{s',s_{i+1},\ldots,s_{m+1}\} &\text{if }s'\succ^R s_{\lambda}\text{, }i\le m\text{, and }|s_{i}|-|s_{i+1}|\text{ is a period of }s',\\
\{s',s_i, s_{i+1},\ldots,s_{m+1}\} & \text{otherwise}.
\end{cases}$$
Consequently, for every $w\in \esigma^*$, we have $\MinSuf(uv,w)\in \{\MaxSuf^R(u,v)w,\MinSuf(v,w)\}$.
\end{lemma}
\begin{proof}
\cref{obs:minmax} yields $\RMaxSuf(uv)\in \{\RMaxSuf(u,v),\RMaxSuf(v)\}$.
By \cref{cor:char}(\ref{it:rmax}) this is equivalent to $\RMaxSuf(uv)\in \{s', s_\lambda\}$.
Consequently, if $s' \preceq^R s_\lambda$, then $\RMaxSuf(uv)=s_\lambda$
and \cref{cor:char}(\ref{it:rmax}) implies $\Sig(uv)=\Sig(s_\lambda)=\Sig(v)$, as claimed.

Thus, we may assume that $s'\succ^R s_{\lambda}$, and in particular that $s'=\RMaxSuf(uv)$.
Let $s_j\in \Sig(w)$ be the longest suffix in $\Sig(uv)\cap \Sig(v)$ ($\lambda \le j \le m+1$).
By~\cref{cor:char}(\ref{it:suf}), $\Sig(uv)\sub\{s'\}\cup\{s_{j},s_{j+1},\ldots,s_{m+1}\}$.
\cref{lem:lyndon_conc} and the definition the $\Sig(\cdot)$ set in terms of the Lyndon factorization yield that 
the inclusion above is actually an equality. 
Moreover, the definition also implies that $s_j$ is a prefix of $s'$, and thus $j\ge i$. 
If $i = m+1$, this already proves our statement, so in the remainder of the proof we assume $i\le m$.

First, let us suppose that $j\ge i+1$. We shall prove that $j=i+1$ and $|s_{i}|-|s_{i+1}|$ is a period of $s'$.
Let $u'$ be a string such that $s'=u's_j$. 
Note that $v_{i}^{p_i}\ldots v_{j-1}^{p_{j-1}}$ is a border of~$u'$ (as $s_i$ is a border of $s'$),
so $v_{j-1}$ is also a border of $u'$ (because $v_{j-1}$ is a prefix of $s_{j-1}$, which is a prefix of $v_i$).
Moreover, by definition of the $\Sig(uv)$  set, $u'$ must be a power of a Lyndon word.
Lyndon words do not proper borders, so any border of $u'$ must be a power of the same Lyndon word.
Thus, $u'$ is a power of $v_{j-1}$. As $v_i$ is a Lyndon word and a prefix of $u'$, this means that $|v_i|\le |v_{j-1}|$. 
Consequently, $i =j+1$ since $|v_i|>|v_{i+1}|>\ldots > |v_{m}|$.
What is more, as $s_{i+1}$ is a prefix of $v_i$, we conclude that $|v_i|$ is a period of $s'=u's_{i+1}$.
Therefore, $|s_{i}|-|s_{i+1}|=p_i|v_i|$ is also a period of $s'$.

It remains to prove that $j=i$ implies that $|s_{i}|-|s_{i+1}|$ is not a period of~$s'$.
For a proof by contradiction suppose that both $s_i\in \Sig(uv)$ and $|s_{i}|-|s_{i+1}|=p_i |v_i|$ is a period of~$s'$.
Let us define $u'$ so that $s'=u's_i=u'v_i^{p_i}s_{i+1}$. 
As $p_i|v_i|$ is a period of $s'$ and $v_i^{p_i}$ contained in $s'$,
we conclude that $s'$ is a substring of $(v_i^{p_i})^\infty=v_i^\infty$, and consequently $|v_i|$ is also a period of $s'$ and hence a period of $u'$ as well.
However, by definition of the $\Sig(\cdot)$ set, $u'$ is a power of a Lyndon word whose length exceeds $|s_i|$ and thus also $|v_i|$.
This Lyndon word cannot have a proper border, and such a border is induced by period $|v_i|$, a contradiction.

Finally, observe that the second claim easily follows from $\Sig(uv)\sub \Sig(v)\cup\{s'\}$.
\end{proof}

We conclude with two combinatorial lemmas, useful to in determining $\RMaxSuf(u,v)$ for $|u|\le |v|$. 
The first of them is also applied later in \cref{sec:gmsq}.

\begin{lemma}\label{lem:lcp}
Let $v\in \Sigma^+$ and $w,w'\in \esigma^+$ be strings such that $w\prec w'$ and the longest common prefix of $w$
and $w'$ is not a proper substring of $v$.
Also, let $\Sig(v)=\{s_{\lambda},\ldots,s_{m-1}\}$.
If $\MinSuf(v,w)=s_i w$, then $\MinSuf(v,w')\in \{s_{i-1}w',s_{i}w'\}$.
\end{lemma}
\begin{proof}
Due to the characterization in \cref{lem:char}, we may equivalently prove that $\rank_{X(v)}(w')$ is $\rank_{X(v)}(w)$ or $\rank_{X(v)}(w)+1$.
Clearly, $\rank_{X(v)}(w) \le \rank_{X(v)}(w')$, so it suffices to show that $\rank_{X(v)}(w') \le \rank_{X(v)}(w)+1$.
This is clear if $|X(v)|=1$, so we assume $|X(v)|>1$. 

This assumption in particular yields that $X'(v)$ consists of proper substrings of $v$, and thus $\rank_{X'(v)}(w)=\rank_{X(v)}(w')$ by the 
condition on the longest common prefix of $w$ and~$w'$. 
However, the inequality in \cref{lem:char} implies 
$$\rank_{X(v)}(w') \le \rank_{X'(v)}(w') = \rank_{X'(v)}(w)\le \rank_{X(v)}(w)+1.$$
This concludes the proof.
\end{proof}

\begin{lemma}\label{lem:2ndleft}
Let $v\in \Sigma^+$, $v=v_1^{p_1}\cdots v_m^{p_m}$ be the Lyndon factorization of $v$, and let $\Sig(v)=\{s_{\lambda},\ldots,s_{m+1}\}$.
If for some $w\in \esigma^*$ and $\lambda< i \le m+1$ we have $\MinSuf(v,w)=s_{i} w$, then $v_{i-1}s_{i} w\preceq sw$ for every non-empty suffix $s$ of $v$ satisfying $|s|>|s_{i}|$.
\end{lemma}
\begin{proof}
Let $s'$ be a string such that $s=s's_{i}$. 
First, suppose that $|s|<|v_{i-1}s_{i}|$.  In this case $s'$ is a proper suffix of a Lyndon word $v_{i-1}$,
and thus $s' \succ v_i$ and, moreover, $sw \succ s' \succ v_{i-1}s_{i}w$. 
Thus, we may assume that $|s|> |v_{i-1}s_{i}|$.

Let $w'=v_{i-1}s_{i} w$ and let $v'$ be a string such that $v=v'v_{i-1}s_{i}$.
Observe that it suffices to prove that $\MinSuf(v',w')=w'$, which implies that $sw \preceq w'$ for $|s|> |v_{i-1}s_{i}|$.
If $v'=\eps$ there is nothing to prove, so we shall assume $|v'|>0$.
Note that we have the Lyndon factorization $v'=v_1^{p_1}\cdots v_{i-1}^{p_{i-1}-1}$ with $i>1$ or $p_{i-1} > 1$.
By \cref{lem:char}, $\MinSuf(v,w)=s_{i}w$ implies $w\prec x_{i-1}^\infty$ and $\MinSuf(v',w')=w'$ is equivalent to
$w' \prec v_{i-1}^\infty$ (if $p_{i-1} > 1$) or $w' \prec v_{i-2}^\infty$ (if $p_{i-1}=1$).
We have $$w' = v_{i-1} s_{i} w \prec v_{i-1} s_{i} x_{i-1}^\infty = v_{i-1} s_{i}(y_{i-1}s_{i})^\infty=v_{i-1}(s_{i}y_{i-1})^\infty=v_{i-1} v_{i-1}^\infty = v_{i-1}^\infty$$ as claimed.
If $p_{i-1}>1$, this already concludes the proof, and thus we may assume that $p_{i-1}=1$.
By definition of the Lyndon factorization we have $v_{i-2}\succ v_{i-1}$, and by \cref{fct:3} this implies $v_{i-2}\succ v_{i-1}^\infty$.
Hence, $w' \prec v_{i-1}^\infty \prec v_{i-2} \prec v_{i-2}^\infty$, which concludes the proof.
\end{proof}

\section{Answering \MSQ}\label{sec:msq}
In this section we present our data structure for \MSQ. 
We proceed in three steps improving the query time from $\Oh(\log |v|)$ via $\Oh(\log^* |v|)$ to $\Oh(1)$.
The first solution is an immediate application of \cref{obs:minmax}(\ref{it:minmax:b}) and the notion of significant suffixes.
Efficient computation of these suffixes, also used in the construction of further versions of our data structure,
is based on \cref{lem:extend}, which yields a recursive procedure. 
The only ``new'' suffix needed at each step is determined using the following result,
which can be seen as a cleaner formulation of Lemma~14 in~\cite{Babenko2016}.
\begin{lemma}\label{lem:max}
Let $u=T[\ell..r]$ and $v=T[r+1..r']$ be fragments of $T$ such that $|u|\le |v|$.
Using the enhanced suffix array of $T$ we can compute  $\MaxSuf^R(u,v)$ in $\Oh(1)$ time.
\end{lemma}
\begin{proof}
Let $sv = \MaxSuf^R(u,v)$ and note that, by \cref{obs:minmax}(\ref{it:minmax:d}), $sv\dol = \MinSuf(u,v\dol)$. 
Let us focus on determining the latter value.
The enhanced suffix array lets us compute a index $k$, $\ell \le k \le r$, which minimizes $T[k..]$.
Equivalently, we have $T[k..]=\MinSuf(u,T[r+1..])$. 
Consequently, $T[k..r] = s_{i}\in \Sig(u)$ for some $\lambda\le i \le m+1$. 
Since $|u|\le |v|$, $v$ is not a proper substring of $u$, and by \cref{lem:lcp}, we have $s\in \{s_{i-1},s_{i}\}$ (if $i=\lambda$, then $s=s_i$).

Thus, we shall generate a suffix of $s_{i-1}$ equal to $s_{i-1}$ if $i>\lambda$, and return the better of the two candidates for $\MinSuf(u,v\dol)$.
If $k=\ell$, we must have $i=\lambda$ and there is nothing to do. Hence, let us assume $k>\ell$.
By \cref{lem:2ndleft}, if we compute an index $k'$, $\ell \le k' < k$, which minimizes $T[k'..]$,
we shall have $T[k'..k-1]=u_{i-1}$ provided that $i>\lambda$. Now, $p_{i-1}$ can be generated as the largest integer such that $u_{i-1}^{p_{i-1}}$ is a suffix of $T[\ell..k-1]$, and we have $|s_{i-1}|=|s_{i}|+p_{i-1}|u_{i-1}|$, which lets us determine $s_{i-1}$.
\end{proof}

\begin{lemma}\label{lem:gensig}
Given a fragment $v$ of $T$, we can compute $\Sig(v)$ in $\Oh(\log |v|)$ time using the enhanced suffix array of $T$
\end{lemma}
\begin{proof}
If $|v|=1$, we return $\Sig(v)=\{v,\eps\}$. Otherwise, we decompose $v=uv'$ so that $|v'|=\big\lceil\frac12|v|\big\rceil$.
We recursively generate $\Sig(v')$ and use \cref{lem:max} to compute $s=\MaxSuf^R(u,v')$.
Then, we apply the characterization of \cref{lem:extend} to determine $\Sig(v)=\Sig(uv')$, using the enhanced suffix array (\cref{thm:esa})
to lexicographically compare fragments of $T$.  

We store the lengths of the significant suffixes in an ordered list. 
This way we can implement a single phase (excluding the recursive calls) in time proportional to $\Oh(1)$ plus the number of suffixes removed from $\Sig(v')$ to obtain $\Sig(v)$.
Since this is amortized constant time, the total running time becomes $\Oh(\log |v|)$ as announced.
\end{proof}

\begin{corollary}\label{cor:logn}
\MSQ queries can be answered in $\Oh(\log|v|)$ time using the enhanced suffix array of $T$.
\end{corollary}
\begin{proof}
Recall that \cref{obs:minmax}(\ref{it:minmax:b}) yields $\MinSuf(v)=\MinSuf(v[1..m-1],v[m])$ where $m=|v|$. Consequently, 
$\MinSuf(v)= sv[m]$ for some $s\in \Sig(v[1..m-1])$. We apply \cref{lem:gensig} to compute $\Sig(v[1..m-1])$
and determine the answer among $\Oh(\log |v|)$ candidates using lexicographic comparison of fragments,
provided by the enhanced suffix array (\cref{thm:esa}).
\end{proof}

\subsection{$\Oh(\log^* |v|)$-time \MSQ}\label{sec:logstar}
An alternative $\Oh(\log |v|)$-time algorithm could be developed based just on the second part of \cref{lem:extend}:
decompose $v=uv'$ so that $|v'|> |u|$ and return $\min(\RMaxSuf(u,v'),\MinSuf(v'))$.
The result is $\MinSuf(v)$ due to \cref{lem:extend} and \cref{obs:minmax}(\ref{it:minmax:b}). 
Here, the first candidate $\RMaxSuf(u,v')$ is determined via \cref{lem:max}, while the second one using a recursive call. 
A way to improve query time to $\Oh(1)$ at the price of $\Oh(n\log n)$-time preprocessing
is to precompute the answers for \emph{basic} fragments, i.e., fragments whose length is a power of two.
Then, in order to determine $\MinSuf(v)$, we perform just a single step of the aforementioned procedure, making sure that $v'$
is a basic fragment. 
Both these ideas are actually present in~\cite{Babenko2016}, along with a smooth trade-off between their preprocessing and query times.

Our $\Oh(\log^* |v|)$-time query algorithm combines recursion with preprocessing for certain \emph{distinguished} fragments.
More precisely, we say that $v=T[\ell..r]$ is distinguished if both $|v|=2^q$ and $f(2^q) \mid r$ for some positive integer $q$,
where $f(x)=2^{\floor{\log\log x}^2}$.
Note that the number of distinguished fragments of length $2^q$ is at most $\frac{n}{2^{\floor{\log q}^2}}=\Oh(\frac{n}{q^{\omega(1)}})$. 

The query algorithm is based on the following decomposition ($x > f(x)$ for $x> 2^{16}$):%
\begin{fact}\label{fct:dec}
 Given a fragment $v=T[\ell..r]$ such that $|v|> f(|v|)$, we can in constant time decompose $v=uv'v''$ such that $1\le |v''|\le f(|v|)$, $v'$ is distinguished, and $|u|\le |v'|$.
\end{fact}
\begin{proof}
Let $q=\floor{\log|v|}$ and $q'=\floor{\log q}^2$.
We determine $r'$ as the largest integer strictly smaller than $r$ divisible by $2^{q'}=f(|v|)$. 
By the assumption that $|v|>2^{q'}$, we conclude that $r'\ge  r-2^{q'}\ge \ell$.
We define $v''=T[r'+1..r]$ and partition $T[\ell..r']=uv'$ so that $|v'|$ is the largest possible power of two.
This guarantees $|u|\le |v'|$.
Moreover, $|v'|\le |v|$ assures that $f(|v'|)\mid f(|v|)$, so $f(|v|') \mid r'$, and therefore $v'$ is indeed distinguished.
\end{proof}

\cref{obs:minmax}(\ref{it:minmax:c}) implies that $\MinSuf(v)\in \{\MinSuf(uv',v''),\MinSuf(v'')\}$.
\cref{lem:extend} further yields $\MinSuf(v)\in \{\MaxSuf^R(u,v')v'',\MinSuf(v',v''),\MinSuf(v'')\}$.
In other words, it leaves us with three candidates for $\MinSuf(v)$. 
Our query algorithm obtains $\MaxSuf^R(u,v')$ using \cref{lem:max}, computes $\MinSuf(v'')$ recursively, 
and determines $\MinSuf(v',v'')$ through the characterization of \cref{lem:char}.
The latter step is performed using the following component 
based on a fusion tree, which we build for all distinguished fragments.

\begin{lemma}\label{lem:fus}
Let $v=T[\ell..r]$ be a fragment of $T$.
There exists a data structure of size $\Oh(\log |v|)$ which answers the following queries
in $\Oh(1)$ time: given a position $r'>r$ compute $\MinSuf(v,T[r+1..r'])$.
Moreover, this data structure can be constructed in $\Oh(\log |v|)$ time using the enhanced suffix array of $T$. 
\end{lemma}
\begin{proof}
By \cref{lem:char}, we have $\MinSuf(v,w)=s_{m+1-\rank_{X(v)}(w)}w$, so in order to determine $\MinSuf(v,T[r+1..r'])$,
it suffices to store $\Sig(v)$ and efficiently compute $\rank_{X(v)}(w)$ given $w=T[r+1..r']$.
We shall reduce these $\rank$ queries to $\rank$ queries in an integer set $R(v)$.

\begin{claim}
Denote $X(v)=\{x_{\lambda}^\infty,\ldots,x_{m}^{\infty}\}$ and let
$$R(v) = \{r + \lcp(T[r+1..],x_{j}^\infty) : x_{j}^\infty \in X(w) \wedge x_{j}^\infty \prec T[r+1..]\}.$$
For every index $r'$, $r < r' \le n$, we have $\rank_{X(v)}(T[r+1..r'])=\rank_{R(v)}(r').$
\end{claim}
\begin{proof}
We shall prove that  for each $j$, $\lambda \le j \le m$, we have
$$x_j^\infty \prec T[r+1..r'] \;\Longleftrightarrow\; \big(r+\lcp(T[r+1..],x_j^\infty)< r'  \;\wedge\; x_j^\infty \prec T[r+1..]\big).$$
First, if $x_j^\infty \succ T[r+1..]$, then clearly $x_j^\infty \succ T[r+1..r']$ and both sides of the equivalence are false.
Therefore, we may assume $x_j^\infty \prec T[r+1..]$. 
Observe that in this case $d:=\lcp(T[r+1..],x_j^\infty)$ is strictly less than $n-r$, and $T[r+1..r+d] \prec x_j^\infty \prec T[r+1..r+d+1]$. Hence, $x_j^\infty \prec T[r+1..r']$ if and only if $r+d< r'$, as claimed.
\end{proof}

We apply \cref{thm:fus} to build a fusion tree for $R(v)$, so that the ranks are can be obtained in $\Oh(1+\frac{\log |R(v)|}{\log W})$ time, 
which is $\Oh(1+\frac{\log \log |v|}{\log \log n})=\Oh(1)$ by \cref{obs:simple}.

The construction algorithm uses \cref{lem:gensig} to compute $\Sig(v)=\{s_{\lambda},\ldots,s_{m+1}\}$. 
Next, for each $j$, $\lambda\le j \le m$,
we need to determine $\lcp(T[r+1..],x_j^\infty)$. This is the same as $\lcp(T[r+1..],(x_j^{p_j})^\infty)$
and, by \cref{obs:simple}, $x_j^{p_j}$ can be retrieved as the suffix of $v$ of length $|s_i|-|s_{i+1}|$. 
Hence, the enhanced suffix array can be used to compute these longest common prefixes and therefore to construct
$R(v)$ in $\Oh(|\Sig(v)|)=\Oh(\log |v|)$ time.
\end{proof}

With this central component we are ready to give a full description of our data structure.

\begin{theorem}\label{thm:logstar}
For every text $T$ of length $n$ there exists a data structure of size $\Oh(n)$ which answers \MSQ in $\Oh(\log^* |v|)$ time
and can be constructed in $\Oh(n)$ time.
\end{theorem}
\begin{proof}
Our data structure consists of the enhanced suffix array (\cref{thm:esa}) and the components
of \cref{lem:fus} for all distinguished fragments of $T$.
Each such fragment of length $2^q$ contributes $\Oh(q)$ to the space consumption and to the construction time,
which in total over all lengths sums up to $\Oh(\sum_{q}\frac{nq}{q^{\omega(1)}})=\Oh(\sum_{q}\frac{n}{q^{\omega(1)}})=\Oh(n)$.

Let us proceed to the query algorithm.
Assume we are to compute the minimal suffix of a fragment $v$.
If $|v|\le f(|v|)$ (i.e., if $|v| \le 2^{16}$), we use the logarithmic-time query algorithm given in \cref{cor:logn}. 
If $|v|>2^{q}$, we apply \cref{fct:dec} to determine a decomposition $v=uv'v''$, which gives us three candidates for $\MinSuf(v)$.
As already described, $\MinSuf(v'')$ is computed recursively, $\MinSuf(v',v'')$ using \cref{lem:fus},
and $\RMaxSuf(u,v')v''$ using \cref{lem:max}. The latter two both support constant-time queries, so the overall time complexity is proportional
to the depth of the recursion. We have $|v''|\le f(|v|)<|v|$, so it terminates.
Moreover,
$$f(f(x))=2^{\floor{\log(\log f(x))}^2}\le 2^{(\log(\log \log x)^2)^2}=2^{4(\log\log\log x)^2}=2^{o(\log\log x)}=o(\log x).$$
Thus, $f(f(x))\le \log x$ unless $x=\Oh(1)$. 
Consequently, unless $|v|=\Oh(1)$, when the algorithm clearly needs constant time, 
the length of the queried fragment is in two steps reduced from $|v|$ to at most $\log |v|$.
This concludes the proof that the query time is $\Oh(\log^* |v|)$.
\end{proof}

\subsection{$\Oh(1)$-time \MSQ}\label{sec:constant}
The $\Oh(\log^* |v|)$ time complexity of the query algorithm of \cref{thm:logstar} is only due to the recursion,
which in a single step reduces the length of the queried fragment from $|v|$ to $f(|v|)$ where $f(x)=2^{\floor{\log \log x}^2}$.
Since $f(f(x))=2^{o(\log\log x)}$, after just two steps the fragment length does not exceed $f(f(n)) = o(\frac{\log n}{\log \log n})$.
In this section we show that the minimal suffixes of such short fragments can precomputed in a certain sense,
and thus after reaching $\tau=f(f(n))$ we do not need to perform further recursive calls.

For constant alphabets, we could actually store all the answers for all $\Oh(\sigma^\tau)=n^{o(1)}$ strings of length up to $\tau$.
Nevertheless, in general all letters of $T$, and consequently all fragments of $T$, could even be distinct. 
However, the answers to \MSQ actually depend only on the relative order between letters, which is captured by 
order-isomorphism.

Two strings $x$ and $y$ are called \emph{order-isomorphic}~\cite{DBLP:journals/ipl/KubicaKRRW13,DBLP:journals/tcs/KimEFHIPPT14}, denoted as $x\approx y$,
if $|x|=|y|$ and for every two positions $i,j$ ($1\le i, j \le |x|$) we have $x[i] \prec x[j] \Longleftrightarrow y[i] \prec y[j].$
Note that the equivalence extends to arbitrary corresponding fragments of $x$ and $y$,
i.e., $x[i..j]\prec x[i'..j'] \Longleftrightarrow y[i..j]\prec y[i'..j']$.
Consequently, order-isomorphic strings cannot be distinguished using \MSQ or \GMSQ.

Moreover, note that every string of length $m$ is order-isomorphic to a string over an alphabet $\{1,\ldots,m\}$.
Consequently, order-isomorphism partitions strings of length up to $m$ into $\Oh(m^m)$ equivalence classes.
The following fact lets us compute canonical representations of strings whose length is bounded by $m=W^{\Oh(1)}$.

\begin{fact}\label{fct:oid}
For every fixed integer $m=W^{\Oh(1)}$, there exists a function $\oid$ mapping each string $w$ of length up to $m$ to a non-negative integer $\oid(w)$ with $\Oh(m\log m)$ bits, so that $w\approx w' \Longleftrightarrow \oid(w)=\oid(w')$. 
Moreover, the function can be evaluated in $\Oh(m)$ time.
\end{fact}
\begin{proof}
To compute $\oid(w)$, we first build a fusion tree storing all (distinct) letters which occur in $w$.
Next, we replace each character of $w$ with its rank among these letters. We allocate $\ceil{\log m}$
bits per character and prepend such a representation with $\ceil{\log m}$ bits encoding $|w|$.
This way $\oid(w)$ is a sequence of $(|w|+1)\ceil{\log m}=\Oh(m\log m)$ bits.
Using \cref{thm:fus} to build the fusion tree, we obtain an $\Oh(m)$-time evaluation algorithm.
\end{proof}

To answer queries for short fragments of $T$, we define overlapping \emph{blocks} of length $m=2\tau$:
for $0\le i \le \frac{n}{\tau}$ we create a block $T_i = T[1+i\tau..\min(n,(i+2)\tau)]$.
For each block we apply \cref{fct:oid} to compute the identifier $\oid(T_i)$. 
The total length of the blocks is bounded $2n$, so this takes $\Oh(n)$ time. 
The identifiers use $\Oh(\frac{n}{\tau}\tau \log \tau)=O(n\log \tau)$ bits of space.

Moreover, for each distinct identifier $\oid(T_i)$, we store the answers to all the \MSQ queries in $T_i$. 
This takes $\Oh(\log m)$ bits per answer, and $\Oh(2^{\Oh(m\log m)}m^2\log m)=2^{\Oh(\tau\log \tau)}$ 
in total.
Since $\tau = o(\frac{\log n}{\log \log n})$, this is $n^{o(1)}$.
The preprocessing time is also $n^{o(1)}$.

It is a matter of simple arithmetic to extend a given fragment $v$ of $T$, $|v|\le \tau$, to a block~$T_i$.
We use the precomputed answers stored for $\oid(T_i)$ to determine the minimal suffix of $v$.
We only need to translate the indices within $T_i$ to indices within $T$ before we return the answer.
The following theorem summarizes our contribution for short fragments:

\begin{theorem}\label{thm:small}
For every text $T$ of length $n$ and every parameter $\tau=o(\frac{\log n}{\log\log n})$ there exists a data structure of size $\Oh(\frac{n\log \tau}{\log n})$ which
can answer in $\Oh(1)$ time \MSQ for fragments of length not exceeding $\tau$.
Moreover, it can be constructed in $\Oh(n)$ time.
\end{theorem}

As noted at the beginning, this can be used to speed up queries for arbitrary fragments:
\begin{theorem}\label{thm:main}
For every text $T$ of length $n$ there exists a data structure of size $\Oh(n)$ which 
can be constructed in $\Oh(n)$ time and answers \MSQ in $\Oh(1)$ time.
\end{theorem}

\section{Answering \GMSQ}\label{sec:gmsq}
In this section we develop our data structure for \GMSQ. 
We start with preliminary definitions and then we describe the counterparts of the three data structures presented in \cref{sec:msq}.
Their query times are $\Oh(k^2\log |v|)$, $\Oh(k^2\log ^*|v|)$, and $\Oh(k^2)$, respectively, i.e., there is an $\Oh(k^2)$ overhead compared
to \MSQ.

We define a \emph{$k$-fragment} of a text $T$ as a concatenation $T[\ell_1..r_1]\cdots T[\ell_k..r_k]$ of $k$ fragments of the text $T$. 
Observe that a $k$-fragment can be stored in $\Oh(k)$ space as a sequence of pairs $(\ell_i,r_i)$.
If a string $w$ admits such a decomposition using $k'$ ($k'\le k$) substrings, we call it a \emph{$k$-substring} of $T$.
Every $k'$-fragment (with $k'\le k$) whose value is equal to $w$ is called an \emph{occurrence} of $w$ as a $k$-substring of $T$.
Observe that a substring of a $k$-substring $w$ of $T$ is itself a $k$-substring of $T$.
Moreover, given an occurrence of $w$, one can canonically assign each fragment of $w$ to a $k'$-fragment of $T$ ($k'\le k$).
This can be implemented in $\Oh(k)$ time and referring to $w[\ell..r]$ in our algorithms, we assume that such an operation
is performed.

Basic queries regarding $k$-fragments easily reduce to their counterparts for 1-fragments:
\begin{observation}\label{obs:esa}
The enhanced suffix array can answer queries (\ref{it:cmp}), (\ref{it:lcps}), and (\ref{it:lcpex}), in $\Oh(k)$ time if $x$ and $y$ are $k$-fragments of $T$.
\end{observation}

\noindent
\GMSQ can be reduced to the following auxiliary queries:
\begin{problem}{\AMSQ}
Given a fragment $v$ of $T$ and a $k$-fragment $w$ of $T$, compute $\MinSuf(v,w)$ (represented as a $(k+1)$-fragment of $T$).
\end{problem}

\begin{lemma}\label{lem:red}
For every text $T$, the minimal suffix of a $k$-fragment $v$ can be determined by $k$ \AMSQ (with $k'< k$)
and additional $\Oh(k^2)$-time processing using the enhanced suffix array of $T$.
\end{lemma}
\begin{proof}
Let $v=v_1\cdots v_k$. 
By \cref{obs:minmax}(\ref{it:minmax:c}), $\MinSuf(v)=\MinSuf(v_k)$ or for some $i$, $1\le i < k$, we have $\MinSuf(v)=\MinSuf(v_i,v_{i+1}\cdots v_k)$. Hence, we apply \AMSQ to determine $\MinSuf(v_i,v_{i+1}\cdots v_k)$ for each $1\le i < k$.
\cref{obs:minmax}(\ref{it:minmax:b}) lets reduce computing $\MinSuf(v_k)$ to another auxiliary query.
Having obtained $k$ candidates for $\MinSuf(v)$, we use the enhanced suffix array to return the smallest among them using $k-1$ comparisons, each performed in $\Oh(k)$ time; see \cref{thm:esa,obs:esa}.
\end{proof}

\begin{fact}\label{fct:gensimple}
\AMSQ can be answered in $\Oh(k\log|v|)$ time using the enhanced suffix array of $T$.
\end{fact}
\begin{proof}
We apply \cref{lem:gensig} to determine $\Sig(v)$, and then we compute the smallest string among $\{sw : s\in \Sig(v)\}$.
These strings are $(k+1)$-fragments of $T$ and thus a single comparison takes $\Oh(k)$ time using the enhanced suffix array.
\end{proof}

\begin{corollary}
\GMSQ can be answered in $\Oh(k^2\log|v|)$ time using the enhanced suffix array of $T$.
\end{corollary}

\subsection{$\Oh(k\log^*|v|)$-time \AMSQ}
Our data structure closely follows its counterpart described in \cref{sec:logstar}.
We define distinguished fragments in the same manner and provide a recursive algorithm based on \cref{fct:dec}.
However, for each distinguished fragment instead of applying \cref{lem:fus}, we build the following much stronger data structure.
Its implementation is provided in Section~\ref{sec:genfus}.
\begin{restatable}{lem}{lemgenfus}\label{lem:genfus}
Let $v$ be a fragment of $T$. There exists a data structure of size $\Oh(\log^2 |v|)$ which answers the following queries in $\Oh(k)$ time:
determine $\MinSuf(v,w)$ for a given $k$-fragment $w$ of $T$.
The data structure can be constructed in $\Oh(\log^2 |v|)$ time; it assumes the access
to the enhanced suffix array of $T$.
\end{restatable}

If $f(|v|)\ge |v|$ ($|v|\le 2^{16}$), we use \cref{fct:gensimple} to compute $\MinSuf(v,w)$ in $\Oh(k\log |v|)=\Oh(k)$ time.
Otherwise, we apply \cref{fct:dec} to decompose $v=uv'v''$ so that $v'$ is distinguished, $|u|\le |v'|$, and $|v''|\le f(|v|)$, where $f(x)=2^{\floor{\log \log x}^2}$.
The characterization of \cref{obs:minmax,lem:extend} again gives three candidates for $\MinSuf(v,w)$:
$\MaxSuf^R(u,v')v''w$, $\MinSuf(v',v''w)$, and $\MinSuf(v'',w)$.
We determine the first using \cref{lem:max}, the second using \cref{lem:genfus}, while the third is computed recursively.
The application of \cref{lem:genfus} takes $\Oh(k+1)$ time, since $v''w$ is a $(k+1)$-fragment of $T$.
We return the best of the three candidates using the enhanced suffix array to choose it in $\Oh(k)$ time.
Since $f(f(x))=o(\log x)$, the depth of the recursion is $\Oh(\log^* |v|)$.
This concludes the proof of the following result:

\begin{theorem}\label{thm:genlogstar}
For every text $T$ of length $n$ there exists a data structure of size $\Oh(n)$ which answers \AMSQ
in $\Oh(k\log^* |v|)$ time and \GMSQ in $\Oh(k^2\log^* |v|)$ time. The data structure 
can be constructed in $\Oh(n)$ time.
\end{theorem}

\subsubsection{Rank queries in a collection of fragments}
The crucial tool we use in the proof of \cref{lem:genfus} is a data structure constructed for a collection $A$ of $W^{\Oh(1)}$ fragments of $T$ to support $\rank_{A}(w)$ queries for arbitrary $k$-fragments $w$ of $T$. 
Since it heavily relies on the compressed trie of fragments in $A$, we start by recalling several related concepts.

A trie is a rooted tree whose nodes correspond to prefixes of strings in a given family of strings $A$.
If $\nu$ is a node, the corresponding prefix $v$ is called the \emph{value} of the node.
The node whose value is $v$ is called the \emph{locus} of $v$.

The parent-child relation in the trie is defined so that the root is the locus of $\eps$,
while the parent $\nu'$ of a node $\nu$ is the locus of the value of $\nu$ with the last character removed.
This character is the \emph{label} of the edge from $\nu'$ and $\nu$.
In general, if $\nu'$ is a ancestor of $\nu$, then label of the path from $\nu'$ to $\nu$ is the concatenation
of edge labels on the path.

A node is \emph{branching} if it has at least two children and \emph{terminal} if its value belongs to $A$.
A \emph{compressed trie} is obtained from the underlying trie by dissolving all nodes except
the root, branching nodes, and terminal nodes. Note that this way we compress paths of vertices with single children,
and thus the number of remaining nodes becomes bounded by $2|A|$.
In general, we refer to all preserved nodes of the trie as \emph{explicit} (since they are stored explicitly)
and to the dissolved ones as \emph{implicit}.
Edges of a compressed trie correspond to paths in the underlying tree
and thus their labels are strings in $\Sigma^+$. 
Typically, these labels are stored as references to fragments of the strings in $A$.

Before we proceed with ranking a $k$-fragment in a collection of fragments, 
let us prove that fusion trees make it relatively easy to rank a suffix in a collection of $W^{\Oh(1)}$ suffixes. 
\begin{fact}\label{fct:one}
Let $A$ be a set of $W^{\Oh(1)}$ suffixes of $T$.
There exists a data structure of size $\Oh(|A|)$,
which answers the following queries in $\Oh(1)$ time: given a suffix $v$ of $T$, find a suffix
$u\in A$ maximizing $\lcp(u,v)$.
The data structure can be constructed in $\Oh(|A|)$ time; it assumes the access to the enhanced suffix array of $T$.
\end{fact}
\begin{proof}
Let $A=\{T[\ell_1..],\ldots,T[\ell_m..]\}$. We build a fusion tree storing $\{ISA[\ell_i] : 1 \le i \le m\}$
and during a query for $v=T[\ell..]$, we determine the predecessor and the successor of $ISA[\ell]$.
We use the $SA$ table to translate these integers into indices $\ell_{i_p}$ and $\ell_{i_s}$.
Since the order of $ISA[\ell_i]$ coincides with the lexicographic order of suffixes $T[\ell_i..]$,
the suffixes $T[\ell_{i_p}..]$ and $T[\ell_{i_s}..]$ are the predecessor $\pred_{A}(v)$ and the successor $\suc_{A}(v)$, respectively.
These are the two candidates for $u\in A$ maximizing $\lcp(u,v)$.
We perform two longest common prefix queries and return the candidate for which the obtained value is larger.
\end{proof}

\begin{lemma}\label{lem:k}
Let $A$ be a set of $W^{\Oh(1)}$ fragments of $T$.
There exists a data structure of size $\Oh(|A|^2)$,
which answers the following queries in $\Oh(k)$ time: given a $k$-fragment $v$ of $T$, determine $\rank_{A}(v)$.
The data structure can be constructed in $\Oh(|A|^2)$ time; it assumes the access
to the enhanced suffix array of $T$.
\end{lemma}
\begin{proof}
Let $A=\{T[\ell_1..r_1],\ldots,T[\ell_m..r_m]\}$ and let $\T$ be the compressed trie of fragments in~$A$.
Note that $\T$ can be easily constructed in $\Oh(m\log m)$ time using the enhanced suffix array.
For each edge we store a fragment of $T$ representing its label and for each terminal node its rank in $A$.
Moreover, for each explicit node $\nu$ of $\T$ we store pointers to the first and last (in pre-order) terminal nodes in its subtree
as well as the following two components:
a fusion tree containing the children of $\nu$ indexed by the first character of the corresponding edge label,
and a data structure of \cref{fct:one} for $\{T[\ell_i+d_\nu..]: \ell_i \in L_\nu\}$,
where $d_\nu$ is the (weighted) depth of $\nu$ and $L_\nu$ contains $\ell_i$ whenever the locus of $T[\ell_i..r_i]$ is in the subtree of $\nu$.
Finally, for each $\ell_i$ we store a fusion tree containing (pointers to) all explicit nodes of $\T$ which represent prefixes of $T[\ell_i..]$, indexed by their (weighted) node depths.
All these components can be constructed in $\Oh(m^2)$ time, with \cref{thm:fus} applied to build fusion trees.

Let us proceed to the description of a query algorithm. Let $v=v_1\cdots v_k$ be the decomposition of the given $k$-fragment
into $1$-fragments, and let $p_i = v_1\cdots v_i$ for $0\le i \le k$.
We shall scan all $v_i$ consecutively and after processing $v_i$, store a pointer to the (possibly implicit) node $\nu_i$
defined as the locus of the longest prefix of $p_i$ present in~$\T$.
We start with $p_0=\eps$ whose locus is the root of $\T$. Therefore, it suffices to describe how to determine $\nu_i$
provided that we know $\nu_{i-1}$. 

If $\nu_{i-1}$ is at depth smaller than $|p_{i-1}|$, there is nothing to do, since $\nu_i = \nu_{i-1}$.
Otherwise, we proceed as follows:
Let $\nu$ be the nearest explicit descendant of $\nu_{i-1}$ ($\nu=\nu_{i-1}$ if $\nu_{i-1}$ is explicit), and let $u$ be a fragment of $T$ representing the label from $\nu_{i-1}$ to $\nu$. 
First, we check if $u$ is a proper prefix of $v_i$. 
If not, $\nu_{i}$ is on the same edge of $\T$ and its depth $|p_{i-1}|+|\lcp(u,v_i)|$. 
Thus, we may assume that $u$ is a proper prefix of $v_i$. Let $v_i = uT[\ell..r]$.  
We make a query to the data structure of \cref{fct:one} built for $\nu$ with $T[\ell..]$ as the query suffix. 
This lets us determine an index $\ell_j \in L_{\nu}$ such that $\lcp(T[\ell..], T[\ell_j+d_\nu..])$ is largest possible. 
This is also an index $\ell_j\in L_{\nu}$ which maximizes $D := \lcp(p_i, T[\ell_j..])=d_\nu + \lcp(T[\ell..], T[\ell_j+d_\nu..])$.
Consequently, $\nu_i$ represents a prefix of $T[\ell_j..]$ and the depth of $\nu_i$ does not exceed $D$.
Thus, the nearest explicit ancestor of $\nu_i$ can be retrieved from the fusion tree built for $\ell_j$ as the node whose depth $D'$ is the predecessor of $D$. If $D' < |p_i|$, we check if that explicit node has an outgoing edge whose label starts with $p_i[D'+1]=T[\ell+D'+1-d_{\nu}]$.
If not, $\nu_{i}$ is equal to the explicit node.
Otherwise, $\nu_{i}$ is an implicit node on the found edge and its depth can be determined using a single longest common prefix query.

After processing the whole $k$-fragment $v$ we are left with $\nu_k$ which is the locus of the longest prefix $p$ of $v$ present in $\T$.
First, suppose that $p\ne v$ and let $c=v[|p|+1]$. Note that by definition of $\nu_k$, this node does not have an outgoing edge labeled with $c$.
If $\nu_k$ has no outgoing edge labeled with a character smaller then $c$,
then the first terminal node of the subtree rooted at the leftmost child of $\nu_k$ represents the successor of $v$ in $A$.
We return its rank as the rank of $v$. Otherwise, we determine the edge going from $\nu_k$ to some node $\nu$ so that the edge label is smaller
than $c$ and largest possible. If $\nu_k$ is explicit, we use the fusion tree to determine $\nu$. We observe that
the predecessor of $v$ in $A$ is the rightmost terminal node in the subtree of $\nu$ and thus we return the rank stored at that node plus one.
Thus, it remains to consider the case when $p=v$. In this case the leftmost terminal node in the subtree of $\nu_k$ is the successor of $v$ in $A$, and thus we return the rank of that node.
\end{proof}

\subsubsection{Proof of \cref{lem:genfus}}\label{sec:genfus}

Having developed the key component, we are ready to generalize \cref{lem:fus}.

\lemgenfus*
\begin{proof}
We use \cref{lem:gensig} to compute $\Sig(v)$ in $\Oh(\log |v|)$ time.
By \cref{lem:char}, in order to find $\MinSuf(v,w)$, it suffices determine $\rank_{X(v)}(w)$.
Moreover, by \cref{lem:char,obs:simple}, $\rank_{X(v)}(w)$ is equal to $\rank_{X'(v)}(w)$ or $\rank_{X'(v)}(w)-1$,
where $X'(v)=\{x_{\lambda}^{p_\lambda},\ldots,x_m^{p_m}\}$ can be determined in $\Oh(\log|v|)$ time from $\Sig(v)$.
We build the data structure of \cref{lem:k} for $A=X'(v)$ so that we can determine $\rank_{X'(v)}(w)$ in $\Oh(k)$ time. 
This leaves two possibilities for $\rank_{X(v)}(w)$, i.e., for $\MinSuf(v,w)$. 
We simply need to compare $s_{i}w$, $s_{i+1}w$ for these two candidates suffixes $s_{i},s_{i+1}\in \Sig(v)$. 
Using the enhanced suffix array, this takes $\Oh(k)$ time.
Consequently, the query algorithm takes $\Oh(k)$ time in total.
In the preprocessing we need to compute $\Sig(v)$ and the data structure of \cref{lem:k} for $A=X'(v)$,
which takes $\Oh(\log |v| + |\Sig(v)|^2)=\Oh(\log^2 |v|)$ time. The space consumption is also $\Oh(\log ^2 |v|)$.
\end{proof}

\subsection{$\Oh(k)$-time \AMSQ}
Like in \cref{sec:constant}, in order to improve the query time in the data structure of \cref{thm:genlogstar}, 
we simply add a component responsible for computing $\MinSuf(v,w)$ for $|v|\le \tau$ where $\tau=f(f(n))=o(\frac{\log n}{\log \log n})$.

Again, we partition $T$ into $\frac{n}{\tau}$ overlapping blocks $T_i$ of length $m=\Oh(\tau)$,
 so that the number of blocks is much larger than the number of order-isomorphism classes of strings of length $\le m$.
Next, we precompute some data for each equivalence class and we reduce a query in $T$ to a query in one of the blocks $T_i$.

While this approach was easy to apply for computing $\MinSuf(v)$ for a fragment $v$ (with $|v|\le \tau$), it is much more difficult for $\MinSuf(v,w)$ for a fragment $v$ ($|v|\le \tau$) and a $k$-fragment $w$. 
That is because $w$ might be composed of fragments $w_j$ starting in different blocks.
As a workaround, we shall replace $w$ by a similar (in a certain sense) $k'$-fragment of $T_i\dol$ ($k'\le k+1$)
where $T_i$ is a block containing $v$.

For $0 \le i < \frac{n}{\tau}$, we define $T_i := T[i\tau+1..\min(n,(i+3)\tau)]$.
We determine $\oid(T_i\$)$ for each block using \cref{fct:oid}.
For each valid identifier we build the enhanced suffix array 
and for all fragment $v$  we construct the set $\Sig(v)$ along with the data structure of \cref{lem:genfus}. 
In total, this data takes $\Oh(2^{\Oh(m\log m)}m^{\Oh(1)})=n^{o(1)}$ space and time to construct.

Now, suppose that we are to compute $\MinSuf(v,w)$ where $|v|\le \tau$ and $w$ is a $k$-fragment of $T$.
We determine the last block $T_i$ containing $v$.
Next, we shall try to represent $w$ as a $k$-fragment of $T_i$.
We will either succeed, or obtain a $k'$-fragment $w'$ of $T_i$ ($k'\le k$)
and a character $c\in \Sigma$ such that $w'c$ is a prefix of $w$ but not a substring of $v$.
In this case \cref{lem:lcp} states that $\MinSuf(v,w'c)$ suffices to determine two candidates for $\MinSuf(v,w)$.

We decompose $w=w_1\cdots w_k$ into fragments and process them iteratively. 
Given a fragment $w_j$ we shall either find an equal fragment of $T_i$
or determine a fragment $w'_j$ of $T_i$ and a character $c\in\Sigma$ such that $w'_j c$ is a prefix of $w_j$ but not a substring of $v$. 
Clearly, if we proceed to $w_{j+1}$ in the first case and terminate in the second,
at the end we successfully represent $w$ or we find a $k'$-fragment $w'=w_1\ldots w_{k'-1}w'_{k'}$ satisfying the desired condition.
Note that since $v$ is a substring of $T[i\tau+1..\min(n,(i+2)\tau)]$, any substring of $v$,
must occur in $T$ at one of the positions in $R_i=\{i\tau+1,\ldots,\min(n,(i+2)\tau)\}$.
Hence, for each block we build a data structure of \cref{fct:one} for suffixes starting in $R_i$.
Given $w_j$ this lets us determine a position $\ell\in R_i$ such that $d_j = \lcp(T[\ell..],w_j)$ is largest possible.
If $d_j=|w_j|$ and $d_j \le \tau$, we have found $w_j$ occurring as a substring of $T_i$.
Otherwise, we set $w'_j = w[1..\min(d_j,\tau)]$, which is a substring of $T_i$, and $c=w[|w'_j|+1]$. 
Clearly, $w'_jc$ is a prefix of $w_j$, so we shall only prove that it is not a substring of $v$.
If $d_j\ge \tau$, then simply $|w'_jc|>\tau\ge |v|$.
Otherwise, by the choice of $\ell$ maximizing $d_j=\lcp(T[\ell..],w_j)$ among $\ell\in R_i$,
the string $w'_jc$ cannot occur at any position in $R_i$ and in particular it cannot be a substring of $v$.

If the described procedure succeeds in finding a $k$-fragment of $T_i$ equal to $w$, we simply apply the data structure of \cref{lem:genfus}
built for $v$ to determine $\MinSuf(v,w)$ in $\Oh(k)$ time. 
Thus, we may assume that this is not the case and it returns a $k'$-fragment $w'$ and a character~$c$.
As already mentioned, having computed $\MinSuf(v,w'c)$, we can determine $\MinSuf(v,w)$ just by
comparing the two candidates with the enhanced suffix array.
If $c$ occurs in $T_i$, then $w'c$ is a $(k'+1)$-fragment of $T_i$ and we may use \cref{lem:genfus} to compute $\MinSuf(v,w'c)$.
Otherwise, we replace $c$ by its successor among letters occurring in $T_i\dol$.
The successor can be computed in constant time provided that for each block we store a fusion
tree of all characters occurring in $T_i\dol$  (mapping each character to a sample position).
To see that replacing $c$ by its successor $c'$ does not change the answer, it is enough
to note that \cref{lem:char} expresses  $\MinSuf(v,w'c)$ in terms of $\rank_{X(v)}(w'c)$,
where $X(v)$ consists of infinite strings composed of characters of $v$ (which are automatically present in $T_i$).

\begin{theorem}\label{thm:gensmall}
For every text $T$ of length $n$ and every parameter $\tau=o(\frac{\log n}{\log\log n})$ there exists a data structure of size $\Oh(n)$ which answers \AMSQ
in $\Oh(k)$ time if $|v|\le \tau$. The data structure can be constructed in $\Oh(n)$ time.
\end{theorem}

This was the last missing ingredient needed to obtain the main result of this paper.

\begin{theorem}\label{thm:genmain}
For every text $T$ of length $n$ there exists a data structure of size $\Oh(n)$ which 
answers \AMSQ in $\Oh(k)$ time and \GMSQ in $\Oh(k^2)$ time. The data structure can be constructed in $\Oh(n)$ time.
\end{theorem}

\section{Applications}\label{sec:app}
As already noted in \cite{Babenko2016}, \MSQ queries can be used to compute Lyndon factorization.
For fragments of $T$, and in general $k=\Oh(1)$, we obtain an optimal solution:

\begin{corollary}\label{cor:lyndon}
For every text $T$ of length $n$ there exists a data structure of size $\Oh(n)$ which
given a $k$-fragment $v$ of $T$ determines the Lyndon factorization $v=v_1^{q_1}\ldots v_m^{q_m}$
in $\Oh(k^2m)$ time. The data structure takes $\Oh(n)$ time to construct.
\end{corollary}
Our main motivation of introducing \GMSQ, however, was to answer \MRQ,
for which we obtain constant query time after linear-time preprocessing.
This is achieved using the following observation; see~\cite{AlgorithmsOnStrings}:
\begin{observation}
The minimal cyclic rotation of $v$ is the prefix of $\MinSuf(v,v)$ of length~$|v|$.
\end{observation}
\begin{theorem}\label{thm:rot}
For every text $T$ of length $n$ there exists a data structure of size $\Oh(n)$ which
given a $k$-fragment $v$ of $T$ determines the lexicographically smallest cyclic rotation of $v$
in $\Oh(k^2)$ time. The data structure takes $\Oh(n)$ time to construct.
\end{theorem}
Using \MRQ, we can compute the Karp-Rabin fingerprint~\cite{DBLP:journals/ibmrd/KarpR87} of the minimal rotations of a given fragment $v$ of $T$
(or in general, of a $k$-fragment). This can be interpreted as a computing fingerprints up to cyclic equivalence,
i.e., evaluating a function $h$ such that $h(\ell,r)=h(\ell',r')$ if and only if $T[\ell..r]$ and $T[\ell'..r']$ are cyclically equivalent.

Consequently, we are able, for example, to count distinct substrings of $T$ with a given exponent $1+1/\alpha$.
They occur within runs or $\alpha$-gapped repeats,
which can be generated in time $\Oh(n\alpha)$ \cite{B14bis,DBLP:journals/corr/CrochemoreKK15,DBLP:journals/corr/GawrychowskiIIK15} and classified using \MRQ according to the cyclic equivalence class of their period.
For a fixed equivalence class the set of substrings generated by a single repeat can be represented as a 
cyclic interval, and the cardinality of a union of intervals is simple to determine;
see also \cite{Extracting_TCS}, where this approach was used to count and list squares and, in general, substrings with a given exponent 2 or more.

\subparagraph*{Acknowledgements}
I would like to thank the remaining co-authors of \cite{Babenko2016}, 
collaboration with whom on earlier results about minimal and maximal suffixes sparked some of my ideas used in this paper. 
Special acknowledgments to Paweł Gawrychowski and Tatiana Starikovskaya for numerous discussions on this subject.

\bibliographystyle{plain}
\bibliography{minsuf}

\newpage

\end{document}